\documentclass[letterpaper, 10 pt, conference]{ieeeconf} 
\pdfoutput=1 
\IEEEoverridecommandlockouts
\usepackage[english]{babel}
\usepackage{cite}
\usepackage{amsmath,amssymb,amsfonts}
\usepackage{algorithmic}
\usepackage{graphicx}
\usepackage{subcaption}
\usepackage{textcomp}
\usepackage{amsmath}
\usepackage{xcolor}
\usepackage{units}
\usepackage{booktabs}
\usepackage{comment}
\usepackage[hidelinks]{hyperref}
\usepackage[normalem]{ulem}

\usepackage{amsthm}
\usepackage{flushend}

\newtheorem{theorem}{Theorem}[section]
\newtheorem{lemma}{Lemma}[section]

\newtheorem{definition}[theorem]{Definition}

\newtheorem{assumption}{Assumption}[section]
\newtheorem{problem}{Problem}[section]

\def\BibTeX{{\rm B\kern-.05em{\sc i\kern-.025em b}\kern-.08em
		T\kern-.1667em\lower.7ex\hbox{E}\kern-.125emX}}
\newif\ifmargincomments 
\margincommentstrue

\ifmargincomments

\else

\fi

\newif\ifextendedversion
\extendedversiontrue

\begin{document}
	
	\title{\Large \bf Urgency-aware Routing in Single Origin-destination Itineraries\\ through Artificial Currencies
	}
	
	\author{Leonardo Pedroso, W.P.M.H. (Maurice) Heemels, Mauro Salazar 
		\thanks{Control Systems Technology section, Eindhoven University of Technology, the Netherlands {\tt \footnotesize \{l.pedroso,w.p.m.h.heemels, m.r.u.salazar\}@tue.nl}}%
	}
	
	\maketitle
	\thispagestyle{plain}
	\pagestyle{plain}
	
	\begin{abstract}
		Within mobility systems, the presence of self-interested users can lead to aggregate routing patterns that are far from the societal optimum which could be achieved by centrally controlling the users' choices. In this paper, we design a fair incentive mechanism to steer the selfish behavior of the users to align with the societally optimal aggregate routing. The proposed mechanism is based on an artificial currency that cannot be traded or bought, but only {spent or received} when traveling. Specifically, we consider a parallel-arc network with a single origin and destination node within a repeated game setting whereby each user chooses from one of the available arcs to reach their destination on a daily basis.
		In this framework, taking faster routes comes at a cost, whereas taking slower routes is incentivized by a reward. The users are thus playing against their future selves when choosing their present actions. To capture this complex behavior, we assume the users to be rational and to minimize an urgency-weighted combination of their immediate and future discomfort.
		To design the optimal pricing, we first derive a closed-form expression for the best individual response strategy. Second, we formulate the pricing design problem for each arc to achieve the societally optimal aggregate flows, and reformulate it so that it can be solved with gradient-free optimization methods.
		Our numerical simulations show that it is possible to achieve a near-optimal routing whilst significantly reducing the users’ perceived discomfort when compared to a centralized optimal but urgency-unaware policy.
	\end{abstract}
	

\section{Introduction}

This paper delves into the challenges confronting present mobility systems, including traffic congestion, environmental pollution, and user dissatisfaction. The advent of cutting-edge technologies such as the internet of things and autonomous driving is ushering in a transformative era in the way we conceptualize mobility, providing an unprecedented opportunity to tackle these challenges. Nonetheless, a fundamental issue is the inherent misalignment between individual objectives, such as minimizing travel time, and societal goals, such as reducing the overall congestion and pollution, which can result in inefficient aggregate routing patterns \cite{RoughgardenTardos2002}.

To address this challenge, this paper employs an incentive scheme, initially proposed in~\cite{CensiBolognaniEtAl2019,SalazarPaccagnanEtAl2021}, based on an artificial currency, Karma, which is designed to align the routing decisions of self-interested users with the socially-optimal aggregate routing while taking into account their temporal individual needs. This innovative framework operates on a currency that can neither be purchased nor traded but can only be gained or expended while traveling. Specifically, we consider a parallel-arc network with a single origin and destination node, which is illustrated in Fig.~\ref{fig:network_example}, whereby each user chooses from one of the available arcs to reach their destination on a daily basis. It empowers users with an equal opportunity to choose between being self-interested and selecting the fastest path for a price or being altruistic and choosing a slower path for a reward.

\begin{figure}[t]
	\centering
	\includegraphics[width = 0.9\linewidth]{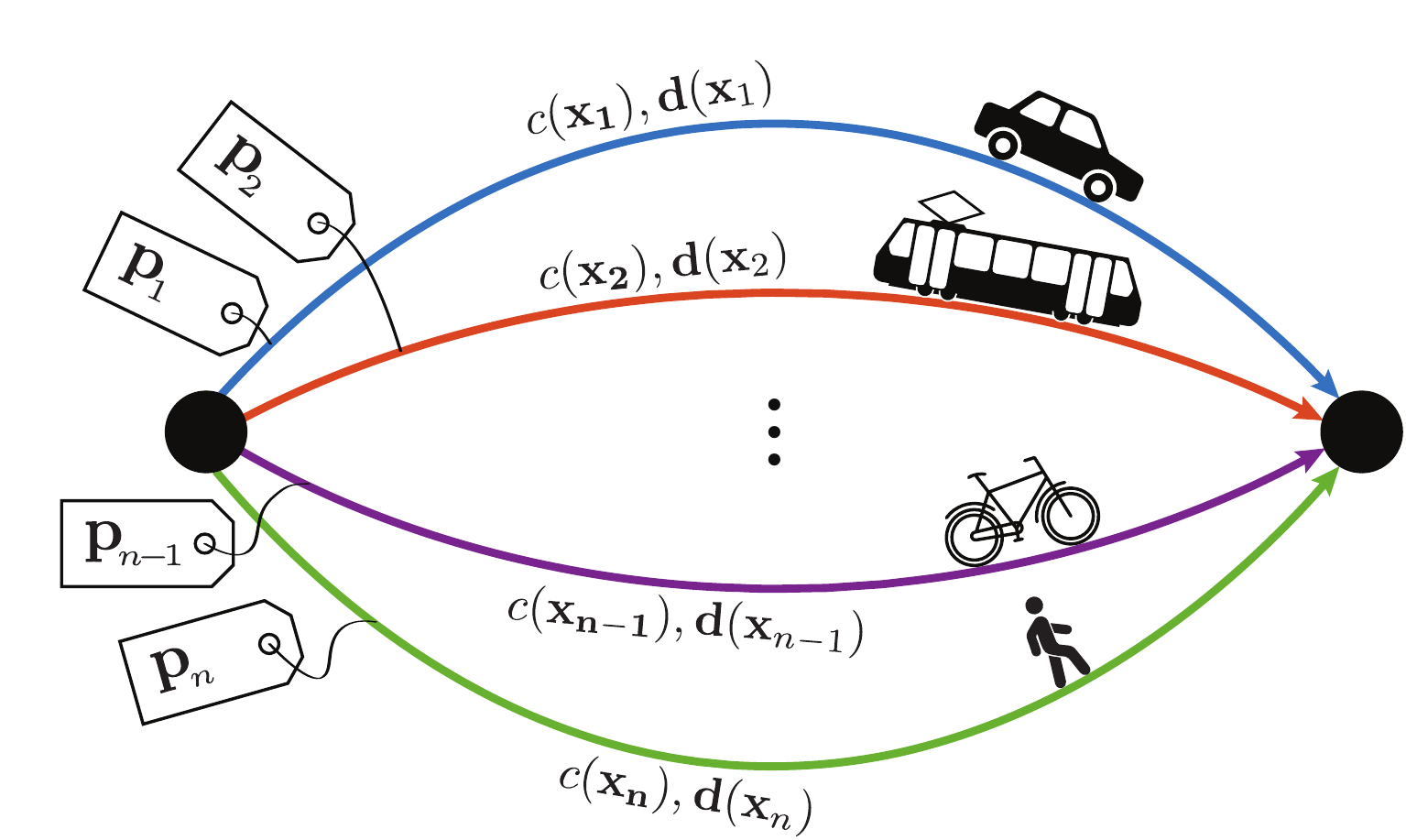}
	\caption{Single origin-destination network of $n$ arcs.}
	\label{fig:network_example}
\end{figure}

\emph{Related work:} 
The toll design problem, which dates back to Pigou's work~\cite{Pigou1920}, has been extensively studied~\cite{Morrison1986,BergendorffHearnEtAl1997}. However, designing tolling mechanisms that account for user sensitivity distribution is a challenging task, for which some results have also been proposed~\cite{FleischerJainEtAl2004,PaccagnanChandanEtAl2019}. Moreover, monetary schemes are intrinsically unfair, as they discriminate against users with lower incomes.

To address this issue, significant attention has been given to the use of artificial currencies to align the aggregate behavior of self-interested users with the system's optimum~\cite{Prendergast2016,GorokhBanerjeeEtAl2019,CensiBolognaniEtAl2019,ElokdaCenedeseEtAl2022}. Nevertheless, these works focus on auction mechanisms forcing users to submit bids every time they desire to use a resource, which may lead to decision fatigue.
Moreover, users are never guaranteed whether they will be able to use the resource or whether they are going to be outbidden. 
Our approach deviates significantly from these mechanisms: We propose simple payment transactions where each itinerary has a fixed cost or reward, thus requiring no bidding whilst accounting for the users' sensitivity by endowing them with freedom of decision as long as they have enough Karma to pay for the desired resource.
In this context, we carried out work for the particular case of two arcs between common origin and destinations nodes in~\cite{SalazarPaccagnanEtAl2021} and, using a reinforcement learning approach, for two and three arcs in~\cite{SandenSchoukensEtAl2023}. Nevertheless, to the best of the author's knowledge, no mechanism has been proposed to cope with the more general $n$ parallel arcs scenario.


\emph{Statement of contributions:} 
The pivotal contributions in this paper are threefold: Focusing on a repeated game setting with $n$ parallel arcs, we first derive a closed-form solution for the best response strategy of a user, enabling an explicit analysis of the repeated game dynamics. Second, we model the aggregate Karma level dynamics under a stationary aggregate routing pattern as an aggregate of Markov chains, bridging the gap between the pricing policy and the aggregate routing pattern. Third, building on the two previous results, we propose a numerical pricing design {procedure driving the aggregate decisions to the system's optimum.} 

\emph{Organization:} This paper is organized as follows. Section~\ref{sec:prob_statement} states the mechanism design problem. In Section~\ref{sec:bestresponse}, the best response strategy of each individual user is analyzed and a closed-form solution is derived, whilst we model the aggregate behavior resulting from the microscopic users' decisions in Section~\ref{sec:mesoscopic}. In Section~\ref{sec:pricing_design}, we devise a pricing design method, whose performance is assessed in Section~\ref{sec:num_res} resorting to numerical simulations. Finally, Section~\ref{sec:concl} presents the main conclusions of this paper.

\emph{Notation:} Throughout this paper, we denote the identity and null matrices, both of appropriate dimensions, by $\mathbf{I}$ and $\mathbf{0}$, respectively. The vectors of ones and zeros, both of appropriate dimensions, are denoted by $\mathbf{1}$ and $\mathbf{0}$. The $i$th component of a vector $\mathbf{v} \in \mathbb{R}^n$ is denoted by $\mathbf{v}_i$. The vector $\mathbf{e_i}$ denotes a column vector whose entries are all set to zero except for the $i$th one, which is set to 1. We denote $[\cdot]_x^y$ as the saturation function with lower bound $x \in \mathbb{R}$ and upper bound $y \in \mathbb{R}$. The cardinality of a set $\mathcal{A}$ is denoted by $|\mathcal{A}|$. The expected value of a random variable $X$ is denoted by~$\mathbb{E}[X]$.


\section{Problem Statement}\label{sec:prob_statement}

This section states the mechanism design problem, which closely follows the formulation in \cite{SalazarPaccagnanEtAl2021}. Three different perspectives are of interest, each corresponding to a subsection: i)~the macroscopic perspective of the central operator, aiming to minimize the societal costs that result from the routing patterns; ii)~the microscopic perspective of each self-interested user, who desires to minimize their daily perceived discomfort;  and iii)~the mesoscopic overarching perspective of the incentive mechanism design framework to align these two seemingly opposing objectives.


Consider the mobility network  with a single origin and destination node connected by $n\in \mathbb{N}$ distinct itineraries depicted in Fig.~\ref{fig:network_example}. We consider a repeated game setting whereby each user chooses from one of the available arcs to reach their destination at each discrete time $t\in \mathbb{N}$.

The incentive mechanism that is employed is based on an artificial currency---Karma. In this framework, taking a particular itinerary comes at a cost of Karma. Let $\mathbf{p}\in \mathbb{R}^n$ denote the prices of the itineraries, i.e., choosing arc $j$ comes at a cost $\mathbf{p}_j$. Users are not allowed to buy or trade Karma, and they can only select arcs that maintain their Karma-level non-negative. It is, hence, crucial that certain (more uncomfortable) arcs are assigned negative prices, which means that users are awarded Karma for taking them.

From a microscopic perspective, denote the route choice of a user $i$ at time instant $t$ by the binary vector $\mathbf{y}^i(t) \in \{0,1\}^n$, whose entry $j$ is given by $\mathbf{y}^i_j(t) = 1$ if the user $i$ chooses arc $j$ at time $t$, and $\mathbf{y}^i_j(t) = 0$ otherwise. Since each user may not travel at time $t$ we can have $\mathbf{y}^i(t) = \mathbf{0}$. Therefore, it follows that $\mathbf{1}^\top\mathbf{y}^i(t) \leq 1$. Let $k^i(t) \in \mathbb{R}_{\geq0}$ denote the amount of Karma that a user $i$ owns at time $t$. Following a routing choice the Karma level is updated according to $k^i(t+1) = k^i(t)-\mathbf{p}^\top\mathbf{y}^i(t)$. 

From a mesoscopic perspective, let $\mathbf{x}(t) \in [0,1]^n$ denote the fraction of users crossing each arc at time $t$. Given a scenario with $M$ users, it is defined as $\mathbf{x}(t) = \frac{1}{M}\sum_{i=1}^M\mathbf{y}^i(t)$. To account for non-traveling users, it is assumed that each user has a constant probability $P_{\text{home}} \in [0,1]$  of not traveling.  Conversely, the probability for a user to travel is $P_{\text{go}}= 1-P_{\text{home}}$ and $\mathbb{E}\left[\mathbf{1}^\top\mathbf{x}(t)\right] = P_{\text{go}}$.

\subsection{Central Operator's Problem}
From the macroscopic perspective of the central operator of the mobility network, the flows across each arc cause a {societal} cost. Let $\mathbf{c}:[0,1]^n \to \mathbb{R}_{\geq 0}^n$ denote the {societal} cost function. It models the {societal} cost of each arc $j$ per user, $\mathbf{c}_j(\mathbf{x}_j(t))$, which is assumed to be monotonically increasing with the fraction  $\mathbf{x}_j(t)$ of users taking it. The desire of the central operator is that the aggregate flows minimize the total {societal} cost {$C(\mathbf{x}):=\mathbf{c}(\mathbf{x})^\top\mathbf{x}$}, which is formulated in the following problem:

\begin{problem}[Central Operator's Problem]\label{prb:design}\color{black}
	The central operator aims at routing customers so that the aggregate flows are
	\begin{equation*}
		\begin{split}
			\mathbf{x}^{\star} \in &\arg \min \limits_{\mathbf{x}\in [0,1]^n}  {C(\mathbf{x})}\\
			& \;\mathrm{s.t.} \phantom{ \min \limits_{x\in [0,1]^n} } \mathbf{1}^\top\mathbf{x} = P_{\mathrm{go}}.
		\end{split}
	\end{equation*}
\end{problem}

\subsection{Individual User's Problem}
From the microscopic perspective of each user, taking an itinerary $j$ comes with a discomfort. Let ${\mathbf{d}:[0,1]^n \to \mathbb{R}_{\geq 0}^n}$ denote the user's discomfort function. It  models the discomfort that stems from taking arc $j$ per user, $\mathbf{d}_j(\mathbf{x}_j(t))$, which is assumed to be monotonically increasing with the fraction of users taking it, $\mathbf{x}_j(t)$.
In contrast to well-known monetary tolling schemes~\cite{BergendorffHearnEtAl1997}, the individual users' complex behavior cannot be captured within a static setting:
From their self-interested view-point, the users are assumed to make choices in order to minimize their traveling discomfort without reaching a negative level of Karma. Thus, the users are playing against their future selves when deciding whether to spend or receive Karma. Furthermore, the perception of discomfort of a user varies daily. The sensitivity to discomfort of a user $i$ at time $t$ is denoted by $s^i(t)$, which is a weighting factor of the daily discomfort.  The sensitivities $s^i(t)$ are assumed to be i.i.d. extractions (w.r.t. $i$ and $t$) of a common probability density function $\rho:[s_{\mathrm{min}},s_{\mathrm{max}}] \to [0,1]$ with support set  $[s_{\mathrm{min}},s_{\mathrm{max}}] \subseteq  \mathbb{R}_{\geq 0}$ and expected value $\bar{s} \in \mathbb{R}_{\geq0}$.  To capture this complex behavior, the users are assumed to be rational and to minimize a combination of their immediate discomfort, weighted by their immediate urgency, and the discomfort encountered for a time period $T$ into the future, weighted by their average urgency. Additionally, we assume each user $i$ to be conservative in terms of Karma, i.e., they will make the route decisions so that their Karma at the end of the horizon will not fall below a reference value ${k_{\mathrm{ref}}^i \in \mathbb{R}_{\geq0}}$. For example, a user may choose a reference value of $\mathbf{p}_1$ to ensure that {they} can still afford to travel in arc $1$ at the end of the horizon. It, thus, depends on the pricing policy $\mathbf{p}$.  Herein, we will assume it to be time-invariant and randomly distributed among the users according to the distribution $\theta_{\mathbf{p}}:\mathbb{R}_{\geq 0} \to\mathbb{R}_{\geq0}$.
{Formally, we obtain} the following individual user's problem:


\begin{problem}[Individual User's Problem]\label{prb:individual}
	At time $t \in \mathbb{N}$, given the flows $\mathbf{x}$ and prices $\mathbf{p}$ a traveling user with Karma level $k\geq 0$, reference $k_{\mathrm{ref}}$, and sensitivity $s$ will choose {their} route as $\mathbf{y}^{\star}$ resulting from
		\begin{subequations}\label{eq:singleAgentAverage}
			\begin{align}\label{eq:singleAgentAverage_cost}
				(\mathbf{y}^\star,\bar{\mathbf{y}}^\star) \in \mathop{\mathrm{argmin}}_{\mathbf{y} \in \mathcal{Y},\;\bar{\mathbf{y}} \in \bar{\mathcal{Y}}} \;&s\, \mathbf{d}(\mathbf{x})^\top\mathbf{y}+ T\,\bar{s}\, \mathbf{d}(\mathbf{x})^\top \bar{\mathbf{y}}\\  \label{eq:singleAgentAverage_c1}
				\mathrm{s.t.}\;\; &k-\mathbf{p}^\top \mathbf{y} - T\mathbf{p}^\top\bar{\mathbf{y}} \geq k_\mathrm{ref}\\ \label{eq:singleAgentAverage_c2}
				&\mathbf{p}^\top \mathbf{y} \leq k,
			\end{align}
		\end{subequations}
	with $ \mathcal{Y} = \{\mathbf{y}\in\{0,1\}^n: \mathbf{1}^\top \mathbf{y} = 1\}$ and $\bar{\mathcal{Y}} =\{\mathbf{y}\in[0,1]^n:\mathbf{1}^\top \mathbf{y}= 1\}$.
	Non-traveling users have $\mathbf{y}^\star = \mathbf{0}$.
\end{problem}

\subsection{Mechanism Design Problem}

Similar to~\cite{SalazarPaccagnanEtAl2021} and following the notation in~\cite{GorokhBanerjeeEtAl2019}, we consider a non-atomic game framework, which corresponds to the limit case where users form a continuum with ${M\to \infty}$. To describe an infinite-user population, let ${\eta_t : \mathbb{R}_{\geq 0}  \times \mathbb{R}_{\geq 0} \to \mathbb{R}_{\geq 0}}$ denote the instantaneous distribution of the Karma level and reference in the population at time $t$, where $\int_0^{\infty} \!\!\int_0^\infty \eta_t (k,k_\mathrm{ref}) \,\mathrm{d}k \, \mathrm{d}k_\mathrm{ref} = 1$. For the infinite-user setting, the Nash and Wardrop Equilibrium (WE) are identical~\cite{PaccagnanGentileEtAl2018} and can be defined as follows:

\begin{definition}[Wardrop Equilibrium]\label{def:WE} \color{black}
	$\mathbf{x}^\mathrm{WE}(t)\in[0,1]^n$ satisfying $\mathbf{1}^\top \mathbf{x}^\mathrm{WE}(t) = P_\mathrm{go}$ is a WE at time $t$, if 
	{\small
		\begin{equation*}
			\begin{split}
				&\mathbf{x}^\mathrm{WE}(t)  =\\
				& \int_{0}^\infty \!\! \!\! \int_{0}^\infty \!\!\!\! 	\int^{s_\mathrm{max}}_{s_\mathrm{min}}\!\!\!\! \mathbf{y}^{\star}\!(\mathbf{x}^\mathrm{WE}(t),s,k,k_\mathrm{ref})\,\rho(s)\,\eta_t(k,k_{\mathrm{ref}})\, \mathrm{d}s\,\mathrm{d}k\,\mathrm{d}k_{\mathrm{ref}},
			\end{split}
	\end{equation*}}%
	where $\mathbf{y}^\star$ is a best response strategy that follows from Problem~\ref{prb:individual}.
\end{definition}


At each time $t$ the aggregate choices of the users are modeled by the WE $\mathbf{x}^{\mathrm{WE}}(t)$. Fig.~\ref{fig:scheme_karma} depicts a scheme of a time-step of the overall model. The mechanism design problem is then to select the arc prices, so that the daily WE converges to the system optimum $\mathbf{x}^\star$, as stated in the problem below:

\begin{figure}[ht]
	\centering
	\includegraphics[width = 0.8\linewidth]{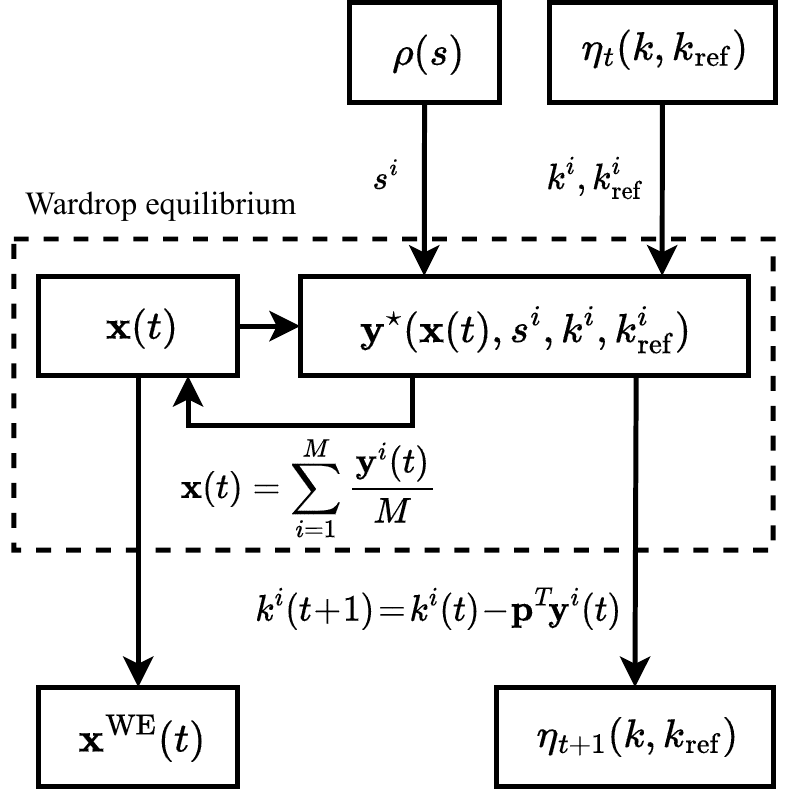}
	\caption{Schematic representation of one time-step of the overall model.}
	\label{fig:scheme_karma}
\end{figure}

\begin{problem}[Pricing Problem]\label{prb:prices}
	Given a desired system optimum $\mathbf{x}^\star$, select $\mathbf{p}\in\mathbb{R}^n$ so that $\lim_{t\to\infty} \mathbf{x}^\mathrm{WE}(t) = \mathbf{x}^\star$.
\end{problem}

To ensure the well-posedness of the pricing problem, the following key assumptions are made on the existence, uniqueness, and convergence of a WE. Given that the best response strategy cannot be formulated in a static setting, and a mixed user strategy is not meaningful, these assumptions are, by no means, obvious. Future research endeavors will focus on the intricacies of the game-theoretic framework, whereas, in this paper, the focus is on the pricing design framework.


\begin{assumption}[Existence and uniqueness of WE]\label{ass:WE}
	Given a Karma level distribution $\eta_t : \mathbb{R}_{\geq 0}  \times \mathbb{R}_{\geq 0} \to \mathbb{R}_{\geq 0}$, a WE $\mathbf{x}^\mathrm{WE}(t)$ exists and is unique.
\end{assumption}

\begin{assumption}[Convergence of WE]\label{ass:WE_convergence}
	For a given pricing strategy $\mathbf{p}$, a stationary WE exists, i.e., ${\mathbf{x}^\mathrm{WE}_{\infty} := \lim_{t\to\infty} \mathbf{x}^\mathrm{WE}(t)}$, irrespective of the initial Karma level distribution $\eta_0 : \mathbb{R}_{\geq 0}  \times \mathbb{R}_{\geq 0} \to \mathbb{R}_{\geq 0}$.
\end{assumption}

\section{Best Response Strategy}\label{sec:bestresponse}

In this section, we focus on the individual user's problem. Specifically, we examine its properties and derive a closed-form solution of the best response strategy, which we will prove to be of paramount importance to the pricing design procedure proposed in this paper. The following result details necessary and sufficient conditions for the feasibility of Problem~\ref{prb:individual}.

\begin{lemma}\label{lem:feasibility}
	Consider a traveling user with Karma $k$, sensitivity $s$, Karma reference $k_\mathrm{ref}$, and prices $\mathbf{p}$. Problem~\ref{prb:individual} is feasible if and only if {\small $k \geq \max(0,k_\mathrm{ref}+(\min_j\mathbf{p}_j)(T+1))$}.
\end{lemma}
\begin{proof}
	The proof can be found in
	\ifextendedversion
	Appendix~\ref{app:proof_feasibility}.
	\else
	the extended version of this paper~\cite{extendedversion}.
	\fi
\end{proof}

{To derive a closed-form solution to Problem~\ref{prb:individual}, we} follow a divide-and-conquer approach. {In a first instance}, in the following theorem, we establish an equivalence between the solutions of Problem~\ref{prb:individual} and a reduced best response problem whose discomforts and prices can be strictly ordered. {More specifically, we make three statements about Problem~\ref{prb:individual}. First, if, for a given arc $j$, there exists an arc $i$ with strictly lower discomfort and cost, then arc $j$ is unreasonable, in the sense that it is never chosen. Second, the discomforts and cost of the reduced set of arcs that are not unreasonable and have distinct discomfort values can be strictly ordered. Third, all integer solutions of Problem~\ref{prb:individual} can be obtained by the solutions to a reduced best response problem, whose discomforts and prices can be strictly ordered. These statements are presented with rigor in the following theorem.}



\begin{theorem}\label{thm:brs_general}
	Consider a traveling user with Karma $k$, sensitivity $s$, and Karma reference $k_\mathrm{ref}$, aggregate flow $\mathbf{x}$, and prices $\mathbf{p}$. Assume, without loss of generality, that the itineraries are numbered, so that $\mathbf{d}_1(\mathbf{x}_1) \leq  \ldots \leq \mathbf{d}_n(\mathbf{x}_n)$ is satisfied. Then, {under the feasibility conditions of Lemma~\ref{lem:feasibility}}:\\
	i) $j^\star \notin \mathcal{J}_u$, where $\mathbf{y}^\star = \mathbf{e_{j^\star}}$ and {\small 	$\mathcal{J}_u := \{ j\in \{1,\ldots,n\} \;|\; \exists i \in  \{1, \ldots n\}: \mathbf{p}_i  \leq  \mathbf{p}_j \land \mathbf{d}_i(\mathbf{x}_i) < \mathbf{d}_j(\mathbf{x}_j) \};$}\\
	ii) {\small $\!\forall i,j  \!\in\!\! \{1,\ldots,n\} \!\setminus \!( \mathcal{J}_u \cup \mathcal{J}_e) \,j\!<\! i  \Longrightarrow \mathbf{p}_j\! > \!\mathbf{p}_i \land \mathbf{d}_j(\mathbf{x}_j)\! > \!\mathbf{d}_i(\mathbf{x}_i)$}, where 	
	{\small \begin{equation*}
			\begin{split}
				&\mathcal{J}_e \!:= \!\left\{ j\!\in \!\{1,\ldots,n\}  \,|\, \exists i \! \in \! \{1, \ldots n\} : \right.\\&\quad \quad \quad \quad\quad 
				\left. \mathbf{d}_i(\mathbf{x}_i) \!= \!\mathbf{d}_j(\mathbf{x}_j) \land (\mathbf{p}_i <\mathbf{p}_j \lor (\mathbf{p}_i  = \mathbf{p}_j  \land i<j) )\right\}.
			\end{split}
	\end{equation*}}\\
	iii) if $(\mathbf{e_{q}}, \mathbf{\bar{y}^{q}})$, $q = 1,\ldots, Q$ are all the solutions to Problem~\ref{prb:individual} for {reduced} aggregate flows $ \{\mathbf{x}_j\}_{j\in \{1,\ldots,n\} \setminus ( \mathcal{J}_u \cup  \mathcal{J}_e) }$, and {reduced} prices $\{\mathbf{p}_j\}_{j\in \{1,\ldots,n\} \setminus ( \mathcal{J}_u \cup  \mathcal{J}_e) }$, then $\mathbf{y}^\star = \mathbf{e_{j^\star}}$ with
	{\small \begin{equation*}
			\begin{split}
				j^\star \in \bigg\{ j\in \{1,\ldots,n\}| \exists q \in \{1,\ldots,Q\} : \mathbf{d}_j(\mathbf{x}_j) = \mathbf{d}_{q^{\star}}(\mathbf{x}_{q^{\star}}) \land  \\ k \geq \mathbf{p}_j \land 
				\left .k - \mathbf{p}_j -T\sum \nolimits_{i\in \{1,\ldots,n\} \setminus ( \mathcal{J}_u \cup  \mathcal{J}_e)} \mathbf{p}_i  \mathbf{\bar{y}^{q}}_i \geq k_{\mathrm{ref}}\right\}
			\end{split}
	\end{equation*}}%
	are all the integer solutions to Problem~\ref{prb:individual}.
\end{theorem}
\begin{proof}
	The proof can be found in
	\ifextendedversion
	Appendix~\ref{app:proof_brs_general}.
	\else
	the extended version of this paper~\cite{extendedversion}.
	\fi
\end{proof}

%

{A few remarks are in order regarding Theorem~\ref{thm:brs_general}. First, depending on the prices and discomforts at a given time, there may be arcs that are not chosen for any sensitivity or Karma level. 
Second, if two arcs have the same discomfort, albeit possibly different prices, both are {equally fit} for a sufficiently high level of Karma. Third, even though similar equivalence conditions could have been stated for the non-integer component of the solutions, they were omitted for the sake of brevity. In a second instance, in the following theorem, a closed-form solution is presented for the aforementioned reduced problems.}



\begin{theorem}\label{thm:brs}
	Consider a traveling user with Karma $k$, sensitivity $s$, and Karma reference $k_\mathrm{ref}$, an aggregate flow $\mathbf{x}$, and prices $\mathbf{p}$.   Assume that  $\mathbf{d}_1(\mathbf{x}_1) <  \ldots < \mathbf{d}_n(\mathbf{x}_n)$ and $\mathbf{p}_1 > \ldots >\mathbf{p}_n$. Let $k(j_1,j_2):= k_{\mathrm{ref}} + \mathbf{p}_{j_1} + T\mathbf{p}_{j_2}$,
	\ifextendedversion
	\begin{equation}\label{eq:j_hat}
		\hat{j}_a := \mathop{\mathrm{argmin}}_{\substack{i\in \{1,\ldots,n\} \setminus \{a\}\\  k \geq \min(k(j,a),k(j,i)) \\ k \leq \max(k(j,a),k(j,i)) }}\frac{\mathbf{d}_i(\mathbf{x}_i)-\mathbf{d}_a(\mathbf{x}_a)}{\mathbf{p}_a-\mathbf{p}_i},
	\end{equation}
	\else
	\begin{equation*}
		\hat{j}_a := \mathop{\mathrm{argmin}}_{\substack{i\in \{1,\ldots,n\} \setminus \{a\}\\  k \geq \min(k(j,a),k(j,i)) \\ k \leq \max(k(j,a),k(j,i)) }}\frac{\mathbf{d}_i(\mathbf{x}_i)-\mathbf{d}_a(\mathbf{x}_a)}{\mathbf{p}_a-\mathbf{p}_i},
	\end{equation*}
	\fi
	\ifextendedversion
	\begin{equation}\label{eq:y_j_a_star}
		\mathbf{\bar{y}}^\star(j,a) := \frac{(k-k(j,\hat{j}_a))\mathbf{e_a} - (k-k(j,a))\mathbf{e_{\hat{j}_a}}}{T(\mathbf{p}_a-\mathbf{p}_{\hat{j}_a})},
	\end{equation}
	\else
	\begin{equation*}
		\mathbf{\bar{y}}^\star(j,a) := \frac{(k-k(j,\hat{j}_a))\mathbf{e_a} - (k-k(j,a))\mathbf{e_{\hat{j}_a}}}{T(\mathbf{p}_a-\mathbf{p}_{\hat{j}_a})},
	\end{equation*}
	\fi
	\ifextendedversion
	\begin{equation}\label{eq:a_hat}
		\hat{a} := \mathop{\mathrm{argmin}} \limits_{a\in \{i,\ldots,n\}} \mathbf{d}(\mathbf{x})^\top	\mathbf{\bar{y}}^\star(j,a),
	\end{equation}
	\else
	\begin{equation*}
		\hat{a} := \mathop{\mathrm{argmin}} \limits_{\substack{a\in \{i,\ldots,n\} \\ k \geq \min(k(j,a),k(j,\hat{j}_a)) \\  k \leq \max(k(j,a),k(j,\hat{j}_a))}} \mathbf{d}(\mathbf{x})^\top	\mathbf{\bar{y}}^\star(j,a),
	\end{equation*}
	\fi
	\begin{equation*}
		\mathbf{\bar{y}_j}^\star := \begin{cases}
			\mathbf{\bar{y}}^\star(j,\hat{a}), \, & k < k(j,1)\\
			\mathbf{e_1}, \, & k \geq k(j,1)
		\end{cases},
	\end{equation*}
	\ifextendedversion
	\begin{equation}\label{eq:gamma_ij}
		\gamma_{i,j} := \begin{cases}
			\frac{T\mathbf{d}^\top (\mathbf{x})\left(	\mathbf{\bar{y}_i}^\star - 	\mathbf{\bar{y}_j}^\star\right)}{\mathbf{d}_j(\mathbf{x}_j)-\mathbf{d}_i(\mathbf{x}_i)}, \; & k \geq \max(0,\mathbf{p}_i,k(i,n))\\
			\infty, & \text{otherwise}
		\end{cases},
	\end{equation}
	\else
	\begin{equation*}
		\gamma_{i,j} := \begin{cases}
			\frac{T\mathbf{d}^\top(\mathbf{x})\left(	\mathbf{\bar{y}_i}^\star - 	\mathbf{\bar{y}_j}^\star\right)}{\mathbf{d}_j(\mathbf{x}_j)-\mathbf{d}_i(\mathbf{x}_i)}, \; & k \geq \min(0,k(i,n))\\
			+\infty, & \text{otherwise}
		\end{cases},
	\end{equation*}
	\fi
\begin{equation*}
		\underline{\gamma}_j := \begin{cases}
			\max_{ i\in \{j+1,\ldots,n\}} \gamma_{j,i} & j <n \\
			-\infty & j = n
		\end{cases} ,
	\end{equation*}
\begin{equation*}
		\bar{\gamma}_j := \begin{cases}
			\min_{i\in \{1, \ldots,j-1\}}  \gamma_{i,j} & j >1 \\
			+\infty & j = 1
		\end{cases},
	\end{equation*}
	\begin{equation*}
		\gamma_{j}(k,k_{\mathrm{ref}},\mathbf{p},\mathbf{d(x)}) := \begin{cases}
			s_{\mathrm{min}}/\bar{s}, \;  & j = n\\
			[\bar{\gamma}_{j+1}]_{s_{\mathrm{min}}/\bar{s}}^{s_{\mathrm{max}}/\bar{s}} ,  \;  &  \bar{\gamma}_{j+1} \geq \underline{\gamma}_{j+1}\\
			\gamma_{j+1},  \; &  \bar{\gamma}_{j+1} < \underline{\gamma}_{j+1}\\
			s_{\mathrm{max}}/\bar{s}, \;  & j = 0,
		\end{cases}
	\end{equation*}
	where the dependence on $k$, $k_{\mathrm{ref}}$, $\mathbf{p}$, and $\mathbf{d(x)}$ was dropped to alleviate the notation. Then,  {under the feasibility conditions of Lemma~\ref{lem:feasibility}}, an optimal response strategy that follows from Problem~\ref{prb:individual} is $\mathbf{y}^\star = \mathbf{e_{j^\star}}${, if and only if} $\bar{\gamma}_{j^\star} \geq \underline{\gamma}_{j^\star}$ and  ${\gamma_{j^{\star}} \leq s/\bar{s} \leq \gamma_{j^\star-1}}$.
\end{theorem}
\begin{proof}
	The proof can be found in
	\ifextendedversion
	Appendix~\ref{app:proof_brs}.
	\else
	the extended version of this paper~\cite{extendedversion}.
	\fi
\end{proof}

\begin{figure}[ht]
	\centering
	\includegraphics[width = 1\linewidth]{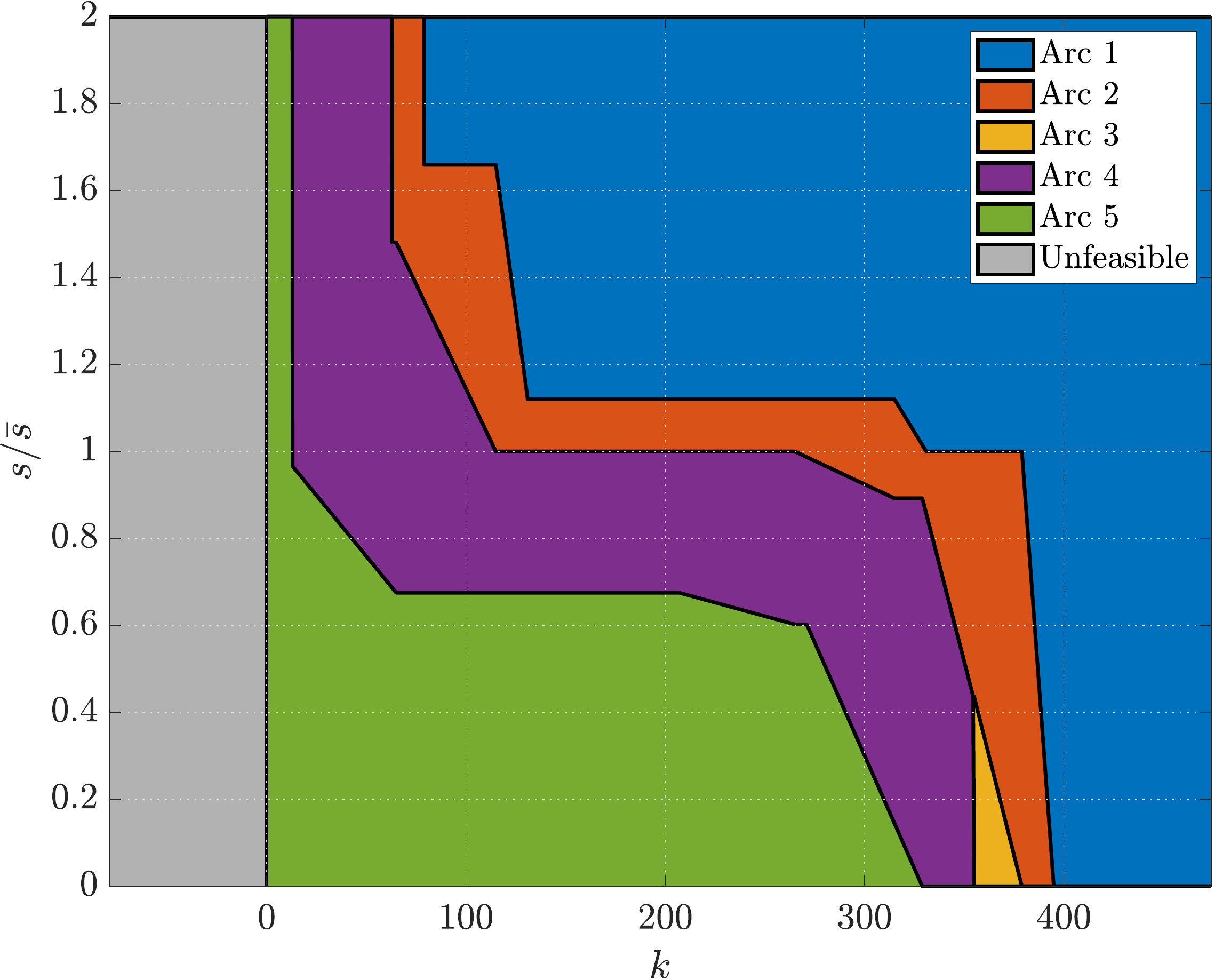}
	\caption{Best response strategy of Problem~\ref{prb:individual} for aggregate flows $\mathbf{x}^\star$, prices $\mathbf{p}^\star$, and $k_\mathrm{ref} = 0$.}
	\label{fig:decision_landscape}
\end{figure}

A few remarks are in order. First, an example of a decision landscape generated by the closed-form solution in Theorems~\ref{thm:brs_general} and \ref{thm:brs} is depicted in Fig.~\ref{fig:decision_landscape}. Second, note that there is an attractive invariant Karma set contained in {\small $[0,  \; k_\mathrm{ref}^i + (T+1)\max_{j}\mathbf{p}_j - \min_j \mathbf{p}_j]$}. Third, the best response strategy is invariant on a positive scaling of $k$, $k_\mathrm{ref}$, and $\mathbf{p}$, i.e., $\mathbf{e_{j^\star}}$ is a best response strategy for $k$, $k_\mathrm{ref}$, and $\mathbf{p}$ {if and only if} it is also for  $\alpha k$, $\alpha k_\mathrm{ref}$, and $\alpha \mathbf{p}$, with $\alpha \in \mathbb{R}_{>0}$. Finally, notice that in contrast to the $n=2$ arcs problem analyzed in \cite{SalazarPaccagnanEtAl2021}, the best response strategy explicitly depends also on the quantitative discomfort levels. 
\section{Mesoscopic Average Behavior}\label{sec:mesoscopic}

Now that we have analyzed the behavior of the individual user's response, we can step back and take a mesoscopic point of view, i.e., model the aggregate behavior resulting from the microscopic decisions. 

\subsection{Aggregate Decision}\label{sec:mesoscopic_aggregate}

At each time $t$, given Karma levels and reference probability distribution $\eta_t(k,k_\mathrm{ref})$, the probability of a traveling user with Karma level $k$ choosing arc $j \in \{1,\ldots,n\}$ is denoted by {$P(j | k,k_\mathrm{ref},\mathbf{p},\mathbf{x}^\mathrm{WE}(t))$} and is, under the conditions of Theorem~\ref{thm:brs}, given by
\begin{equation*}
	P(j|k,k_\mathrm{ref},\mathbf{p},\mathbf{x}^\mathrm{WE}(t)) = \!  \int\limits_{\gamma_{j}\left(k,k_\mathrm{ref},\mathbf{p},\mathbf{d}\left(\mathbf{x}^\mathrm{WE}(t)\right)\right)}^{\gamma_{j-1}\left(k,k_\mathrm{ref},\mathbf{p},\mathbf{d}\left(\mathbf{x}^\mathrm{WE}(t)\right)\right)} \rho(s)\,  \mathrm{d}s\,.
\end{equation*}
Thus, the discrete-time evolution of the Karma level density function can be written as
\begin{equation*}
	\begin{split}
		&\eta_{t+1}(k,k_\mathrm{ref}) =  P_{\mathrm{home}}\eta_t(k,k_\mathrm{ref}) \\
		&+ P_{\mathrm{go}} \sum_{j = 1}^n P\left(j|k+\mathbf{p}_j,k_\mathrm{ref},\mathbf{p},\mathbf{x}^\mathrm{WE}(t)\right) \eta_t(k+\mathbf{p}_j,k_\mathrm{ref}),
	\end{split}
\end{equation*}
for $k, k_\mathrm{ref} \in \mathbb{R}_{\geq 0}$, and $t\in \mathbb{N}$. Moreover, the definition of the WE equilibrium in Definition~\ref{def:WE} can be rewritten as
\begin{equation}\label{eq:WE_cf} \small 
		\mathbf{x}^\mathrm{WE}_j(t) \!=\!
		P_\mathrm{go}\!\!\int\limits_{0}^\infty\! \!\int\limits_{0}^\infty \! \!P(j|k,k_\mathrm{ref},\mathbf{p},\mathbf{x}^\mathrm{WE}(t)) \eta_t(k,k_{\mathrm{ref}})\, \mathrm{d}k\,\mathrm{d}k_{\mathrm{ref}},\!\!
\end{equation}%
$j = 1,\ldots,n$.

It is important to point out two key aspects. First, note that only in the strict ordering conditions of Theorem~\ref{thm:brs}, it is possible to write a closed-form deterministic expression for $P_t(j|k,k_\mathrm{ref},\mathbf{p},\mathbf{d}(\mathbf{x}^\mathrm{WE}(t)))$. If they are not satisfied, it is only known that $P_t(j|k,k_\mathrm{ref},\mathbf{p},\mathbf{d}(\mathbf{x}^\mathrm{WE}(t)))$ is such that the aggregate decisions reconstruct the aggregate flows at the WE, which is portrayed in \eqref{eq:WE_cf}. Second, remark the discrete nature of the evolution of the Karma level density function, which is a linear combination of the previous density function shifted by $n$ fixed values that correspond to the arcs' prices. Thus, although continuous Karma levels were considered up to this point, a user $i$ with a given initial Karma level $k_0$, can only evolve to Karma levels that are of the form $k =k_0 + \sum_{j=1}^n m_j\mathbf{p}_j$ with $m_j \in \mathbb{N}_0$. This observation suggests that modeling the Karma level evolution of a single user as a Markov chain is appropriate.

{\color{black} Although there is a bounded attractive Karma level set, as mentioned earlier, the number of distinct Karma levels cannot be bounded even if $||\mathbf{p}||$ is bounded. Henceforth, to prevent that we consider that {$\mathbf{p}$ is a vector of integers, i.e.,} $\mathbf{p}\in \mathbb{Z}^n$. Nevertheless, it is important to recall that due to the positive scaling invariance of the prices and Karma levels on the user's decision, pointed out in Section~\ref{sec:bestresponse}, the precision of the prices can be chosen to be as high as desired by increasing $||\mathbf{p}||$, amounting to enforce $\mathbf{p}\in \mathbb{Q}^n$ in a computationally tractable manner.}


\subsection{Stationary Markov Chain Model}

Consider a single user $i$ and assume that we are in the strict ordering conditions of Theorem~\ref{thm:brs}. Starting at a Karma level $k_0$, if $\mathbf{d}(\mathbf{x})$ is held constant, it is possible to propagate the possible Karma transitions and generate a finite Markov chain.  Let $\mathcal{K}(k_0,k^i_\mathrm{ref},\mathbf{p},\mathbf{x}) = \{k^i_1, \ldots, k^i_{|\mathcal{K}(k_0,k^i_\mathrm{ref},\mathbf{p},\mathbf{x})|}\}$ denote the state space of the chain and $\mathbf{A}(k_0,k^i_\mathrm{ref},\mathbf{p},\mathbf{x}) \in \mathbb{R}^{|\mathcal{K}(k_0,k^i_\mathrm{ref},\mathbf{p},\mathbf{x})| \times |\mathcal{K}(k_0,k^i_\mathrm{ref},\mathbf{p},\mathbf{x})| }$ the corresponding transition matrix in column-stochastic form, whereby the states are ordered by their corresponding Karma level. For the remainder of this subsection, the dependence of $\gamma_j$, $\mathcal{K}$, and $\mathbf{A}$ on $k_0$, $k^i_\mathrm{ref}$, $\mathbf{p}$, and $\mathbf{x}$ are dropped to alleviate the notation.


The entries of $\mathbf{A}$ can be expressed in closed-form by
\begin{equation*}
	\mathbf{A}_{uv}  = P_\mathrm{home}\mathbf{I} +  \begin{cases}
		0,  &\nexists j: k^i_v\!-\!k^i_u = \mathbf{p}_j\\
		P_\mathrm{go}\int\limits_{\gamma_{j}}^{\gamma_{j-1}} \rho(s)\,  \mathrm{d}s, & \exists j: k^i_v \!-\!k^i_u = \mathbf{p}_j\\
	\end{cases}.
\end{equation*}
Since $P_\mathrm{home}>0$,  the Markov chain is aperiodic. Note, however, that it is not necessarily irreducible, since there may exist more than one communication class. By the Perron-Frobenius Theorem \cite[Theorem~2.12]{Bullo2018}, it follows that the eigenvalue $\lambda=1$ is dominant but not necessarily simple. Denote the eigenvector associated with the eigenvalue $\lambda=1$ that corresponds to the stationary Karma distribution over  $\mathcal{K}$ of the Markov chain initialized in $k_0$ by $\boldsymbol{\pi}_{\infty}(k_0,k^i_\mathrm{ref},\mathbf{p},\mathbf{x})$. Notice that it corresponds to the limit of the power iteration of $\mathbf{A}$ initialized at the Karma level distribution with all probability concentrated in $k_0$. Finally, define the stationary arc selection matrix $\mathbf{P}(k_\mathrm{ref},\mathbf{p},\mathbf{x}) \in \mathbb{R}^{n \times |\mathcal{K}|}$ as the matrix whose entry $(u,v)$ is  given by $\mathbf{P}_{uv}(k_\mathrm{ref},\mathbf{p},\mathbf{x}) = P(u|k^i_v,k_\mathrm{ref},\mathbf{p},\mathbf{x})$. 



\subsection{WE as an Aggregate Markov Chain}

In the previous subsection, we modeled the stationary behavior of a single user under the conditions of Theorem~\ref{thm:brs} as a Markov chain. Now, we analyze the aggregate of the Markov chains that model the stationary behavior of each user. More specifically, given that this model is distinct only for distinct $k_\mathrm{ref}$, the aggregate over the Karma reference distribution is taken. In that regard, on the Assumption~\ref{ass:WE_convergence}, in steady-state, \eqref{eq:WE_cf} can be rewritten as 
\begin{equation}\label{eq:WE_cf_mc}
	\begin{split}
		&\mathbf{x}^\mathrm{WE}_\infty \!=  \\
		&P_\mathrm{go} \!\int\limits_{0}^\infty \! \mathbf{P}(k_\mathrm{ref},\mathbf{p},\mathbf{x}^\mathrm{WE}_\infty)  \boldsymbol{\pi}_{\infty}(k_0,k_\mathrm{ref},\mathbf{p},\mathbf{x}^\mathrm{WE}_\infty)\theta_{\mathbf{p}}(k_{\mathrm{ref}})\, \mathrm{d}k_{\mathrm{ref}}.\!\!\!\!\!\!
	\end{split}
\end{equation}

\section{Pricing Design Problem}\label{sec:pricing_design}

The pricing design problem, formulated in Problem~\ref{prb:prices}, is now tackled on the following assumption:

\begin{assumption}\label{ass:ordering}
	Assume that, at the system optimum, there is an arc ordering such that  $\mathbf{d}_1(\mathbf{x}_1^\star) < \ldots < \mathbf{d}_n(\mathbf{x}_n^\star)$ is satisfied.
\end{assumption}

Under Assumptions~\ref{ass:WE}, \ref{ass:WE_convergence}, and \ref{ass:ordering}, the problem amounts to finding $\mathbf{p} = \mathbf{p}^\star$ such that \eqref{eq:WE_cf_mc} is satisfied for $\mathbf{x}^\mathrm{WE}_\infty = \mathbf{x}^\star$. Notice that without Assumption~\ref{ass:ordering}, neither $\mathbf{P}(k_\mathrm{ref},\mathbf{p},\mathbf{x}^\mathrm{WE}_\infty)$ nor  $\boldsymbol{\pi}_{\infty}(k_0,k_\mathrm{ref},\mathbf{p},\mathbf{x}^\mathrm{WE}_\infty)$ would be deterministic, which follows from the analysis in Section~\ref{sec:mesoscopic_aggregate}.
{It is important to} point out that first, the integer nature of $\mathbf{p}$, i.e. $\mathbf{p} \in \mathbb{Z}^n$, makes it challenging to solve \eqref{eq:WE_cf_mc}.
Second, the Karma reference distribution $\theta_\mathbf{p}(k_\mathrm{ref})$ depends on the pricing policy $\mathbf{p}$.
Third, not only do the entries of $\mathbf{P}(k_\mathrm{ref},\mathbf{p},\mathbf{x}^\mathrm{WE}_\infty)$ and  $\boldsymbol{\pi}_{\infty}(k_0,k_\mathrm{ref},\mathbf{p},\mathbf{x}^\mathrm{WE}_\infty)$ in  \eqref{eq:WE_cf_mc} depend nonlinearly on $\mathbf{p}$, but also the dimensions of the matrix and vector themselves change with $\mathbf{p}$.

To find the optimal prices, we enforce $\mathbf{1}^\top\mathbf{x}^{\star} = P_\mathrm{go}$, which via \eqref{eq:WE_cf_mc} only enforces {one constraint} on $\mathbf{p}$. The additional constraints stem from the fact that, at steady-state, the expected Karma level remains constant, hence $\mathbf{p}^{\star \top}\mathbf{x}^{\star} = 0$, and from the fact that the best response strategy is invariant on a positive scaling of prices and Karma distributions. Whilst these constraints were sufficient to design the optimal static prices for the 2-arc setting~\cite{SalazarPaccagnanEtAl2021}, for the general $n$-arc case under consideration we still need to find the optimal $\mathbf{p}^\star$ satisfying~\eqref{eq:WE_cf_mc} with $\mathbf{x}^\mathrm{WE}_\infty = \mathbf{x}^\star$, which, as mentioned above, is highly nonlinear and non-smooth. To the best of the authors' knowledge, these features make the derivation of a closed-form solution not feasible.

\subsection{Numerical Design Method}

{We leverage the structure of the problem to overcome the aforementioned difficulties and reframe the pricing design problem thoughtfully so that it can be solved efficiently. In this regard, we introduce three considerations to enable the numerical solution of~\eqref{eq:WE_cf_mc} for $\mathbf{p}$ with $\mathbf{x}^\mathrm{WE}_\infty = \mathbf{x}^\star$.} First, $\theta_\mathbf{p}(k_\mathrm{ref})$ has to be bounded and discrete to be numerically tractable. Note that this is a reasonable assumption since there is an attractive invariant Karma set and the Karma levels are discrete because $\mathbf{p}\in \mathbb{Z}^n$. Second, since $\mathbf{p} \in \mathbb{Z}^n$, the equality in \eqref{eq:WE_cf_mc} will not be achieved exactly. Instead, one may attempt to minimize {the deviation of the cost of the right-hand term w.r.t. the optimal aggregate flows.} Nevertheless, the larger $||\mathbf{p}||$ is allowed to be, the closer is the equality. Third, the constraint $(\mathbf{p}^\star)^\top\mathbf{x}^{\star} = 0$ may not be satisfied exactly if the entries of $\mathbf{x}^\star$ are irrational or if $||\mathbf{p}||$ is bounded. Thus, one can substitute it with a quantized  approximation $\mathbf{x}^\star_\mathrm{quant}$.


Therefore, the proposed pricing design optimization problem {becomes:}

\begin{problem}[Numerical Pricing Design]\label{prb:prices_num}
	Given a desired system optimum $\mathbf{x}^\star$, select $\mathbf{p}$ as the solution to
\begin{equation}\label{eq:num_prob}
	\begin{split}
		&	\!\!\!\!\!\!\!{\small \mathop{\mathrm{min}}\limits_{\mathbf{p}\in \mathbb{Z}^n} {C}\!\left(\!\!P_\mathrm{go}\!\!\!\!\!\sum\limits_{k_\mathrm{ref} = k_{\mathrm{ref}_{\mathrm{min}}}}^{k_{\mathrm{ref}_{\mathrm{max}}}}\!\!\!\!\!\!\! \mathbf{P}(k_\mathrm{ref},\mathbf{p},\mathbf{x}^\star) \boldsymbol{\pi}_{\infty}(k_0,k_\mathrm{ref},\mathbf{p},\mathbf{x}^\star)\theta_{\mathbf{p}}(k_{\mathrm{ref}})\! \!\right)}\\
		& \!\!\!\!\mathop{\mathrm{s.t.}}  \quad \mathbf{p}^\top\mathbf{x}^{\star}_\mathrm{quant} = 0\\
		& \!\!\!\! \phantom{\mathop{\mathrm{s.t.}}}\;\quad \mathbf{p}_j > \mathbf{p}_{j+1}, \, j = 1,\ldots, n-1\\
		& \!\!\!\!\phantom{\mathop{\mathrm{s.t.}}}\;\quad \mathbf{p}_1 > 0\\
		& \!\!\!\!\phantom{\mathop{\mathrm{s.t.}}}\;\quad \mathbf{p}_n < 0,
	\end{split}\!\!\!\!\!\!\!\!\!
\end{equation}
where $k_{\mathrm{ref}_{\mathrm{min}}}$ and $k_{\mathrm{ref}_{\mathrm{max}}}$ are the minimum and maximum values of the support of $\theta_\mathbf{p}$, respectively.
\end{problem}

Such a problem can be efficiently solved with gradient-free methods, as shown in Section~\ref{sec:num_res} below. Furthermore, a useful particularity of Problem~\ref{prb:prices_num} is that the minimum of the objective function is known and given by $C(\mathbf{x}^\star)$. Thus, it is easy to evaluate the suboptimality bound and stop the numerical method whenever it reaches a given threshold.
	
%


\section{Numerical Results}\label{sec:num_res}

In this section, numerical results are presented for an illustrative case study with $n = 5$. We consider $M = 1000$ users of which, on average, $P_\mathrm{home} = 5\%$ do not travel every day. Their daily sensitivity is sampled from a uniform distribution on the interval $[0,2]$ and their prediction horizon is $T=4$. We model the discomfort as a travel-time Bureau of Public Roads (BPR) function \cite{BPR1964}
\begin{equation*}
	\mathbf{d}_j(\mathbf{x}_j) = \mathbf{d^0}_{\!\!j}\left(1+ \alpha (\mathbf{x}_j/\boldsymbol{\kappa}_j)^\beta\right),
\end{equation*}
with $\alpha = 0.15$, $\beta = 4$, and $\mathbf{d^0}$ and $\boldsymbol{\kappa}$ were generated randomly which, rounded to four decimal places, are given by $\mathbf{d^0} = [0.5001 \; 0.5734 \; 0.7085 \; 0.6512 \; 0.8602]^\top$ and $\boldsymbol{\kappa} = [0.0923 \; 0.1863 \; 0.3968 \; 0.3456 \;0.5388]^\top$, ordered according to the arc ordering in Assumption~\ref{ass:ordering}. We consider distribution of the reference values $\theta_\mathbf{p}$ to be a discrete uniform distribution with support $\{k_\mathrm{ref}\in \mathbb{N} \,|\, k_\mathrm{ref} = 0 \lor k_\mathrm{ref} = \mathbf{p}_j \}$, which corresponds to users having the possibility of saving Karma to afford traveling through an arc with a positive price at the end of the horizon.  The system's cost is considered to be a weighted sum of the travel-time in each link, i.e., $\mathbf{c}_j(\mathbf{x}):=  \mathbf{c^0}_{\!\!j}\mathbf{d}_j(\mathbf{x}_j)$, whose weights were randomly generated and, rounded to four decimal places, are given by $\mathbf{c^0}_{\!\!\!j} = [0.7096 \; 0.8426 \; 0.9391 \; 0.6022 \; 0.5137]$. This can correspond to the weighted minimization of, for example, sound pollution.

Rounded to four decimal places, {employing \cite{Loefberg2004},} $\mathbf{x}^\star = [0.0877 \; 0.1309 \; 0.0000 \; 0.3053 \; 0.4261]^\top$  and ${\mathbf{d}(\mathbf{x}^\star) = [0.5611 \; 0.5943 \; 0.7085 \; 0.7107 \; 0.9106]^\top}$, which is in accordance with Assumption~\ref{ass:ordering}. The optimization problem \eqref{eq:num_prob} is solved using a standard genetic algorithm method subject to $||\mathbf{p}||_\infty \leq 100$ {in less than 500 wall-clock seconds in a standard laptop}, whose solution is ${\mathbf{p}^\star = [79\;63\;39\;13\;-45 ]^\top}$. We considered $k_0 = \mathbf{p}_1$ and $\mathbf{x}^{\star}_\mathrm{quant}$ {resulting from rounding $\mathbf{x}^\star$ to three decimal places}.

\begin{figure}[ht!]
	\begin{subfigure}{\linewidth}
		\centering
		\includegraphics[width = \linewidth]{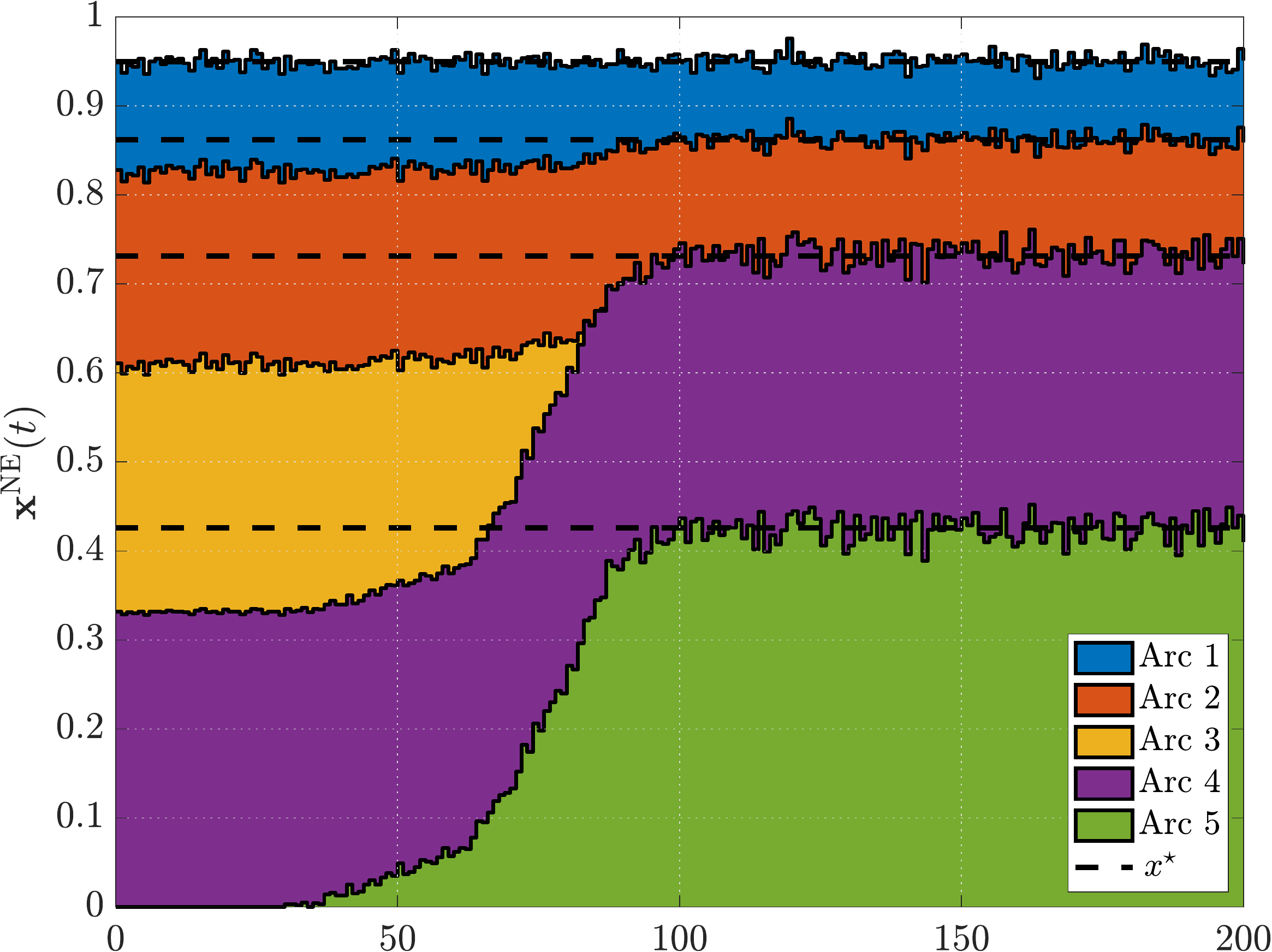}
		\caption{Evolution of aggregate flows.}
		\label{fig:decision}
	\end{subfigure}\\%
	\begin{subfigure}{\linewidth}
		\raggedleft
		\includegraphics[width = 0.99\linewidth]{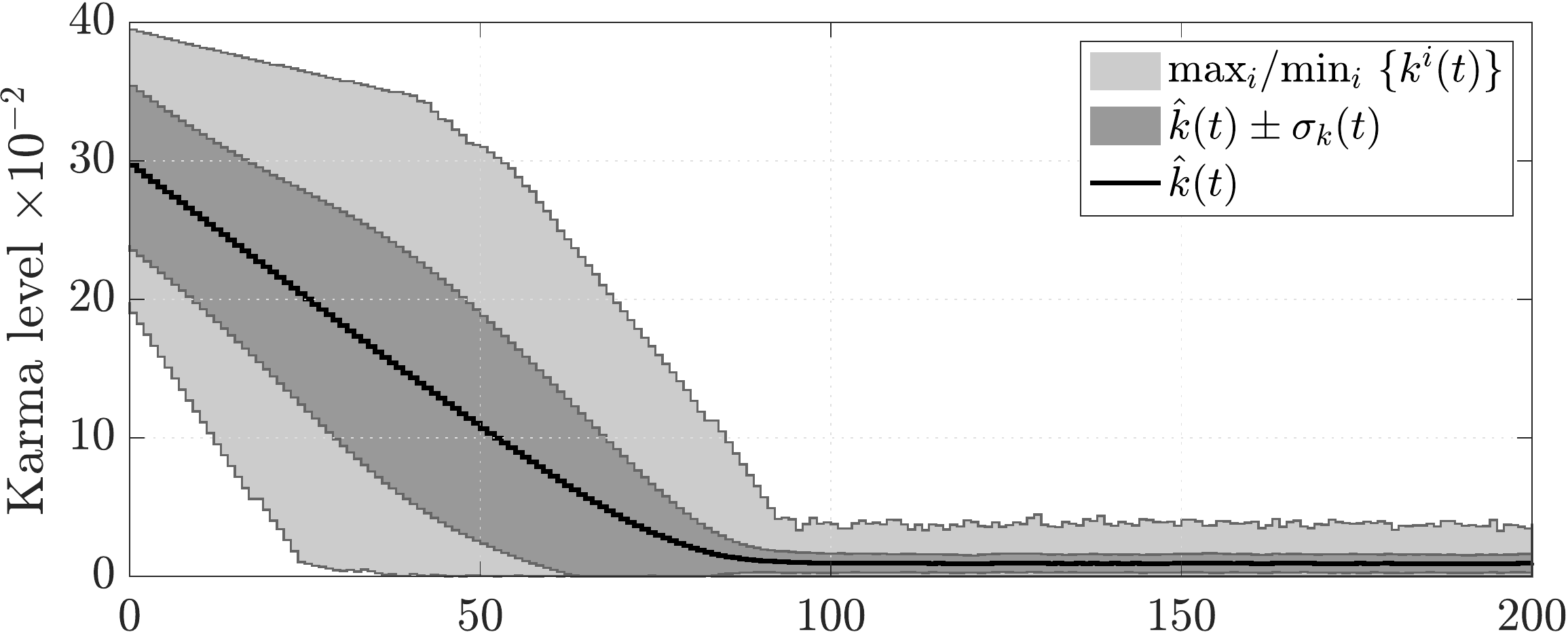}
		\caption{Evolution of Karma level.}
		\label{fig:karma}
	\end{subfigure}\\%
	\begin{subfigure}{\linewidth}
		\raggedleft
		\includegraphics[width = 0.99\linewidth]{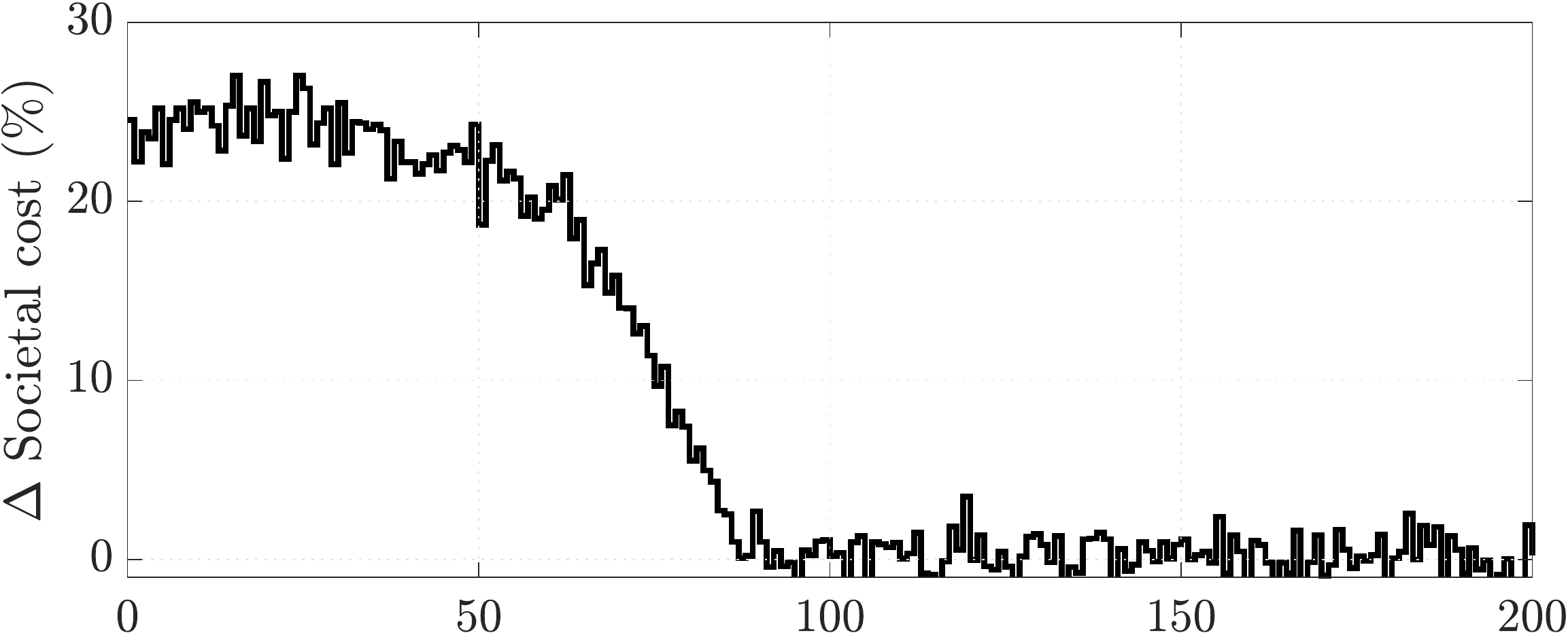}
		\caption{Evolution of the relative cost difference.}
		\label{fig:cost}
	\end{subfigure}\\%
	\begin{subfigure}{\linewidth}
		\raggedleft
		\includegraphics[width = \linewidth]{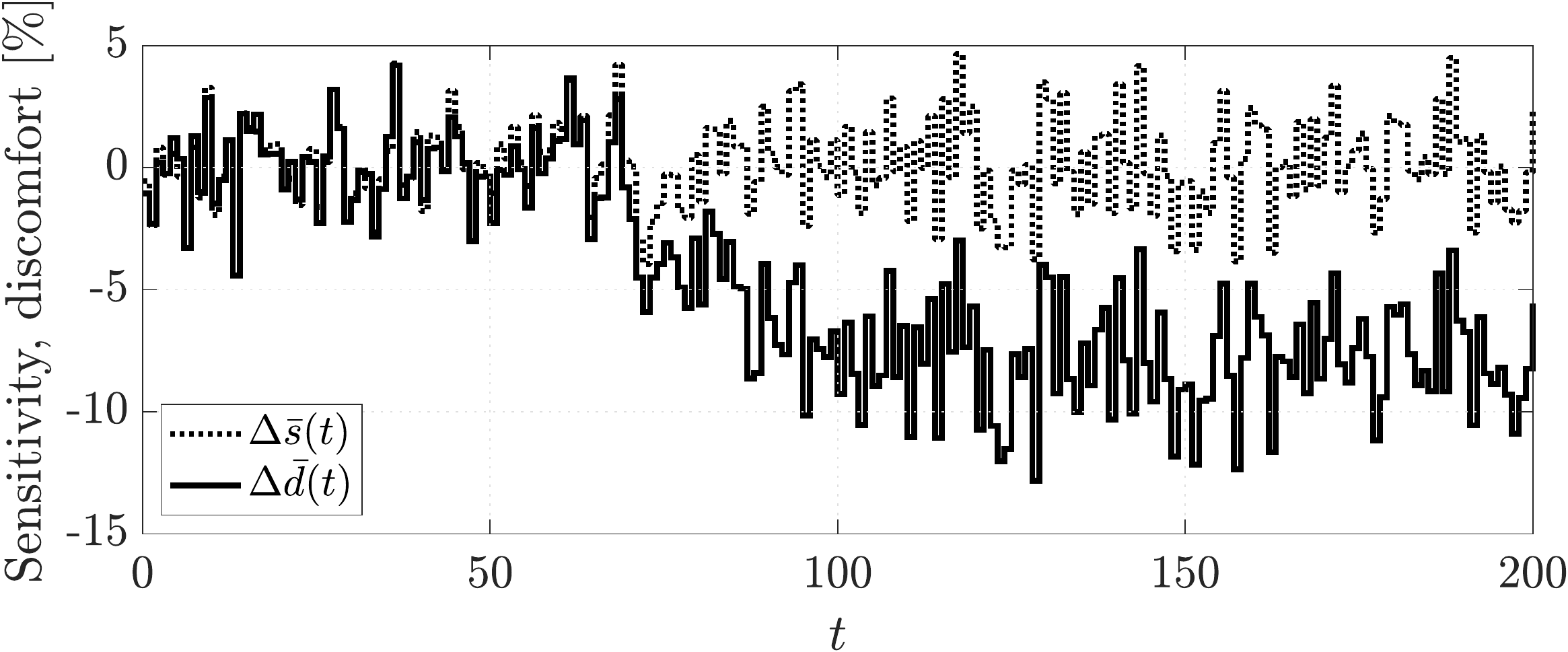}
		\caption{Evolution of the relative sensitivity and discomfort deviation.}
		\label{fig:sensitivity}
	\end{subfigure}%
	\caption{Numerical simulation results.}
	\label{fig:num_sim}
\end{figure}

The daily simulations are carried out by computing the Nash equilibrium that follows from the decisions of each user to Problem~\ref{sec:bestresponse}, which approximate the WE as $M\to\infty$. The Karma values were initialized randomly according to  a discrete uniform distribution with support $\{25\mathbf{p}_1,25\mathbf{p}_1,+1,\ldots,50\mathbf{p}_1\}$.  Figs.~\ref{fig:decision}--\ref{fig:cost} depict the evolution of the aggregate flows, Karma level, and relative cost difference in relation to the system optimum, respectively, throughout the simulation. We denote the average and the standard deviation of the users' Karma level at time $t$ by $\hat{k}(t)$ and $\sigma_k(t)$, respectively. First, since the initial Karma levels are very high, the users act as if the pricing scheme were not implemented. This can be seen in the initial plateau in Fig.~\ref{fig:cost} which is associated with the constant aggregate flows visible in Fig.~\ref{fig:decision}. Nevertheless, as the users' Karma is depleted, as shown in Fig.~\ref{fig:karma}, the users can no longer afford every link and the pricing mechanism drives the aggregate flows to the system-optimal flows. Second, it is important to point out that, despite all the assumptions made to tackle the intractability of the pricing design problem and enable a numerical solution, the prices that were designed get very close to the system optimum, as visible in Fig.~\ref{fig:cost}, with an average relative difference in relation to the theoretical optimum of  $0.15\,\%$ only over the last $50$ instants of the simulation. In fact, the steady-state aggregate flows of the numerical simulation closely match $\mathbf{x}^\star$, as visible in Fig.~\ref{fig:decision}. Third, we analyze  i)~the relative difference of the average perceived discomfort w.r.t. a scenario in which the users are centrally allocated to the optimal flows randomly, i.e. without taking into account their sensitivity, which is given by
\begin{equation*}
	\Delta \bar{d}(t) := \frac{\sum_{i = 1}^M s^i(t) \mathbf{d}(\mathbf{x})^\top\mathbf{y}^i(t) + \bar{s}\, \mathbf{d}(\mathbf{x})^\top\mathbf{y}^i(t) }{\sum_{i = 1}^M \bar{s}\, \mathbf{d}(\mathbf{x})^\top\mathbf{y}^i(t)};
\end{equation*}
and ii)~the relative deviation of the average sensitivity in relation to the expected sensitivity, i.e., ${\Delta \bar{s}(t) := (1/M)\sum_{i = 1}^M (s^i(t)-\bar{s})/\bar{s}}$. Fig.~\ref{fig:sensitivity} depicts the evolution of these two quantities. It is noticeable that, at steady-state, the perceived discomfort is roughly $8\,\%$ lower in comparison to an optimal but urgency-unaware policy.

%
%
%

Due to space limitations, some details regarding the numerical pricing design and the simulation were omitted.  Nevertheless, a MATLAB implementation as well as additional simulation results,  is openly available in an open source repository at  {\small\texttt{\url{https://fish-tue.github.io/single-origin-destination-routing}}}.
\section{Conclusion}\label{sec:concl}
In this paper, we explored a fair incentive mechanism based on artificial currencies to tackle routing problems whilst accounting for the daily urgency of the users. We modeled the system as a repeated game and we obtained a closed-form solution for the user's daily strategy, which enables a numerical solution of the {arc-pricing} design problem. We showed that by employing a simple static payment-transaction scheme, our approach steers the aggreggate flows towards the societally optimal flows, achieving the minimum societal cost. On top of that, the simulation results indicated that the proposed scheme allows for a significant reduction of the users' perceived discomfort in relation to an optimal but urgency-unaware policy.
	
%
%
%

{In the future, we {aim} to apply this scheme to an intermodal mobility network and to electric vehicle charging problems.}



\section*{Acknowledgment}
We thank Dr.\ I.\ New and F.\ Paparella for proofreading the~paper.

\appendices
\section{Proof of Lemma~\ref{lem:feasibility}} \label{app:proof_feasibility}
If {\small $k \geq \max(0,k_\mathrm{ref}+(\min_j\mathbf{p}_j)(T+1))$}, then  $\mathbf{y} = \bar{\mathbf{y}} = \mathbf{e_{\mathop{\mathrm{argmin}}_j(\mathbf{p}_j)}}$ is feasible. Conversely, if the optimization problem is feasible, then there is at least one  pair $(\mathbf{y},\bar{\mathbf{y}})$ that satisfies the constraints. Given that {$ k-(T+1)\min_j\mathbf{p}_j \geq k-\mathbf{p}^\top \mathbf{y} - T\mathbf{p}^\top\bar{\mathbf{y}} \geq k_\mathrm{ref}$}
and $\min_j\mathbf{p}_j \leq \mathbf{p}^\top \mathbf{y} \leq k$,
the reciprocal is true and the result follows immediately.

\section{Proof of Theorem~\ref{thm:brs_general}} \label{app:proof_brs_general}

{Before proceeding with the proof it is worth pointing out the significance of sets $\mathcal{J}_u$ and $\mathcal{J}_e$. First, note that if, for a given arc $j$, there exists an arc $i$ with strictly lower discomfort and cost, then arc $j$ is never an integer solution, because choosing arc $i$ is still feasible and would always achieve lower cost. We denote the set of such unreasonable arcs as  $\mathcal{J}_u$. Second, note that the objective function of Problem~\ref{prb:individual} does not depend on the prices of the chosen arcs. In fact, if two arcs $i$ and $j$ have the same discomfort and both are feasible, then they are equally fit integer solutions. In that regard, one can attempt to consider only a set of arcs with unique discomforts and then extend the solution to the other equally fit arcs that have the same discomfort. In that regard, we define the set $\mathcal{J}_e$ of arcs that have repeated discomforts and have the highest prices. Thus, the set ${\{1,\ldots,n\} \setminus ( \mathcal{J}_u \cup  \mathcal{J}_e)}$ contains the arcs that are not unreasonable and have unique discomforts with the lowest price.}

Statement i) is proved by contradiction. Assume that $(\mathbf{e_{j^\star}},\bar{\mathbf{y}}^\star)$ is a solution with $j^\star \in \mathcal{J}_u$. It follows from the definition of $ \mathcal{J}_u$ that there is an arc $i \in \{1,\ldots,n\}$ such that $\mathbf{p}_i \leq \mathbf{p}_j$ and  $\mathbf{d}_i(\mathbf{x}_i) \!< \!\mathbf{d}_j(\mathbf{x}_j)$. Thus, since $(\mathbf{e_{j^\star}},\bar{\mathbf{y}}^\star)$ is feasible, $(\mathbf{e_i}, \bar{\mathbf{y}}^\star)$ must also be feasible. Additionally,  $s\, \mathbf{d}_i(\mathbf{x}_i)+ T\,\bar{s}\, \mathbf{d}(\mathbf{x})^\top \bar{\mathbf{y}}^\star < s\, \mathbf{d}_j(\mathbf{x}_j)+ T\,\bar{s}\, \mathbf{d}(\mathbf{x})^\top \bar{\mathbf{y}}^\star$, i.e., $(\mathbf{e_i}, \bar{\mathbf{y}}^\star)$ achieves a lower cost than $(\mathbf{e_{j^\star}}, \bar{\mathbf{y}}^\star)$, which is a contradiction.


To prove statement ii), consider $i,j \in \{1,\ldots,n\} : j<i$. Under the discomfort ordering assumption, $\mathbf{d}_i(\mathbf{x}_i) \geq  \mathbf{d}_j(\mathbf{x}_j)$. If $i,j \notin \mathcal{J}_u$ and $\mathbf{d}_i(\mathbf{x}_i) >  \mathbf{d}_j(\mathbf{x}_j)$, then $\mathbf{p}_i < \mathbf{p}_j$. If $i,j \notin \mathcal{J}_e$, then $\mathbf{d}_i(\mathbf{x}_i) >  \mathbf{d}_j(\mathbf{x}_j)$. Thus, if $i,j \notin \mathcal{J}_u \cup \mathcal{J}_e$, then $\mathbf{p}_i < \mathbf{p}_j$.

We now turn to the proof of statement iii). Assume that $(\mathbf{e_{q}}, \mathbf{\bar{y}^{q}})$ for some $q \in \{1,\ldots,Q\}$ is a solution to Problem~\ref{prb:individual} for aggregate flows $\{\mathbf{x}_j\}_{j\in \{1,\ldots,n\} \setminus ( \mathcal{J}_u \cup  \mathcal{J}_e) }$, and prices $\{\mathbf{p}_j\}_{j\in \{1,\ldots,n\} \setminus ( \mathcal{J}_u \cup  \mathcal{J}_e)}$, denoted for the remainder of this proof as the reduced problem. Then, the additional arcs available in Problem~\ref{prb:individual} for aggregate flows $\mathbf{x}$ and prices $\mathbf{p}$, denoted for the remainder of this proof as the original problem,  are those in the set $\mathcal{J}_u\cup\mathcal{J}_e$. Arcs in $\mathcal{J}_u$ were already proved to be unfeasible, according to statement i). Arcs in $\mathcal{J}_e$ have discomforts that are equal to the discomfort of one and only one arc in $\{1,\ldots,n\} \setminus ( \mathcal{J}_u \cup  \mathcal{J}_e)$, which has the lowest price due to the way $\mathcal{J}_e$ is defined. Thus, $(\mathbf{e_{q}}, \mathbf{\bar{y}^{q}})$ is also a solution to the original problem as well as any other arcs in $\mathcal{J}_e$ that achieve the same discomfort, i.e.,  $j\in \mathcal{J}_e: \mathbf{d}_j(\mathbf{x}_j) =\mathbf{d}_q(\mathbf{x}_q)$, and whose prices are still feasible, i.e., $j\in \mathcal{J}_e: k \geq \mathbf{p}_j \land k - \mathbf{p}_j -T\sum \nolimits_{i\in \{1,\ldots,n\} \setminus ( \mathcal{J}_u \cup  \mathcal{J}_e)} \mathbf{p}_i  \mathbf{\bar{y}^{q}}_i \geq k_{\mathrm{ref}}$. It remains to prove that all the solutions to the original problem are obtained, employing this procedure, from the solutions  $(\mathbf{e_{q}}, \mathbf{\bar{y}^{q}})$ with $q\in \{1,\ldots,Q\}$ to the reduced problem. Assume that $(\mathbf{e_{j^\star}},\mathbf{\bar{y}}^\star)$ is a solution to the original problem. It was already proved that $j^\star \notin \mathcal{J}_u$. Then, either $j^\star \in  \{1,\ldots,n\} \setminus ( \mathcal{J}_u \cup  \mathcal{J}_e) $ or $j^\star \in \mathcal{J}_e$. First, assume the former. It is immediate that $(\mathbf{e_{j^\star}},\mathbf{\bar{y}^q})$  is a solution of the reduced problem, where $\mathbf{\bar{y}^{q}}$ is obtained from $\mathbf{\bar{y}}^\star$ by adding together the entries corresponding to arcs with equal discomfort. Second, assume the latter, i.e., $j^\star \in \mathcal{J}_e$. Then, from the definition of $\mathcal{J}_e$, it follows that there is one and only one  $q^\star \in \{1, \ldots n\} \setminus ( \mathcal{J}_u \cup  \mathcal{J}_e) :\mathbf{d}_{q^\star}(\mathbf{x}_{q^\star}) = \mathbf{d}_j(\mathbf{x}_j) \land \mathbf{p}_{q^\star} \leq\mathbf{p}_j$. It follows that $(\mathbf{e_{q}}, \mathbf{\bar{y}^{q}})$ is a solution to the reduced problem, since the condition $\mathbf{p}_{q^\star} \leq\mathbf{p}_j$ ensures that it is feasible, where $\mathbf{\bar{y}^{q}}$ is obtained as previously described.

\section{Proof of Theorem~\ref{thm:brs}} \label{app:proof_brs}

Problem~\ref{prb:design} is a mixed-integer linear programming (MILP) optimization problem. To obtain a closed-form solution, the following procedure is employed: We start by assuming the integer part of the solution is known, i.e, $\mathbf{y}^\star = \mathbf{e_j}$ for some $j\in \{1,\ldots,n\}$, and then we compute the optimal non-integer variables, denoted by $\mathbf{\bar{y}_j^\star}$, assuming that integer decision, which reduces to a linear programming (LP) optimization problem. The solution is, afterwards, given by the pair $(\mathbf{e_j},\mathbf{\bar{y}_j^\star})$ that achieves the lowest cost.

First, assume that $\mathbf{y}^\star = \mathbf{e_j}$ for some $j\in \{1,\ldots,n\}$. Note that all $j$ that do not satisfy $\mathbf{p}_j \leq k \land k \geq k(j,n)$ can be immediately discarded, because at least one of the constraints \eqref{eq:singleAgentAverage_c1} and  \eqref{eq:singleAgentAverage_c2} is not satisfied. The problem is, thus, reduced to an LP given by
\begin{subequations}\label{eq:singleAgentLP}
\begin{align}
	\mathbf{\bar{y}_j^\star} \in \mathop{\mathrm{argmin}}_{\bar{\mathbf{y}}  \in [0,1]^n} \;& \mathbf{d}(\mathbf{x})^\top \bar{\mathbf{y}}\\ \label{eq:singleAgentLP_c1}
	\mathrm{s.t.}\;\; &k-\mathbf{p}_j- T\mathbf{p}^\top\bar{\mathbf{y}} \geq k_\mathrm{ref}\\ \label{eq:singleAgentLP_c2}
	&\mathbf{1}^\top \bar{\mathbf{y}}  =  1  \,.
\end{align}
\end{subequations}
Note that cost function of the LP is simply the inner product of $\mathbf{d}(\mathbf{x})$ and $\bar{\mathbf{y}}$. Introduce constraint $\mathbf{1}^\top\bar{\mathbf{y}} = 1$ to eliminate $\bar{\mathbf{y}}_a$, for some $a\in \{1,\ldots,n\}$, in the cost function and in the inequality \eqref{eq:singleAgentLP_c1}, which yields 
\begin{equation}\label{eq:proof_brs_karma_cost}
\mathbf{d}_a(\mathbf{x}_a) + \sum_{\substack{i = 1\\i\neq a}}^n \bar{\mathbf{y}}_i(\mathbf{d}_i(\mathbf{x}_i)-\mathbf{d}_a(\mathbf{x}_a))
\end{equation}
and
\begin{equation}\label{eq:proof_brs_karma_constr}
\sum_{\substack{i = 1\\i\neq a}}^n \bar{\mathbf{y}}_i(\mathbf{p}_i-\mathbf{p}_a) \leq  \frac{k-k(j,a)}{T},
\end{equation}
respectively. Now we disregard $\bar{\mathbf{y}}_a$ and solve the LP in the $n-1$ dimensional space of the remaining components of $\bar{\mathbf{y}}$. 

We start by noting that there is a solution $\mathbf{\bar{y}^\star_j}$ that has, at most, two non-zero entries. Equivalently, there exists at least one  $a \in \{1,\ldots,n\}$ such that the solution of the $n-1$ dimensional LP is along an axis in a Cartesian frame. This statement can be proved by contradiction. Assume that, for any $a \in \{1,\ldots,n\}$ no solution is along the axes. Consider only one, $a_1 \in \{1,\ldots,n\}$. Because it is an LP optimization problem, a solution must lie in one of the vertices of the  polytope whose faces are defined by $ \bar{\mathbf{y}}_i \geq 0, \forall i \in \{1,\ldots,n\}\setminus \{a_1\}$,  \eqref{eq:proof_brs_karma_constr}, and
\begin{equation}\label{eq:proof_brs_sum_constraint}
\sum_{\substack{i = 1\\i\neq a_1}}^n \bar{\mathbf{y}}_i \leq 1.
\end{equation}
A vertex is, thus, at the intersection of $n-1$ hyperplanes defined by the boundary of these constraints. The only combination of $n-1$ hyperplanes that yields a vertex that is not along the axis is of $n-3$ hyperplanes of the form $ \bar{\mathbf{y}}_i = 0,  i \in \{1,\ldots,n\}\setminus \{a_1\}$ and those defined by the boundaries of \eqref{eq:proof_brs_karma_constr} and \eqref{eq:proof_brs_sum_constraint}. Therefore, \eqref{eq:proof_brs_sum_constraint} is an active constraint of the solution, which is of the form 
\begin{equation}\label{eq:proof_brs_one_sol_contradiction}
\bar{\mathbf{y}}_i = \begin{cases}
	K,  &i = u\\
	1-K,  & i = v\\
	0,  &i\notin \{a_1,u,v\}
\end{cases}, i \in \{1,\ldots,n\}\setminus\{a_1\}
\end{equation}
for some $K\in (0,1)$, and some $u,v \in  \{1,\ldots,n\}\setminus\{a_1\}$. So the solution of the initial $n$ dimensional LP is given by \eqref{eq:proof_brs_one_sol_contradiction} and $\bar{\mathbf{y}}_{a_1} = 0$. Now consider $a_2 = u$. In the new $n-1$ dimensional LP obtained by eliminating $\bar{\mathbf{y}}_{a_2}$, the solution is along the axis of the Cartesian frame, which is a contradiction.

Now, one can choose, for a given and fixed $j$, among all the vertices of the polytope along the axes for all  $a \in \{1,\ldots,n\}$ the one that yields the lowest cost, which is guaranteed to be a solution. For a given $a$ there can be three different types of vertices: at the origin and at the intersection of the boundary hyperplane of either \eqref{eq:proof_brs_karma_constr} or  \eqref{eq:proof_brs_sum_constraint} with the axes. The intersection of the boundary hyperplane \eqref{eq:proof_brs_sum_constraint} with an axis of coordinate $\bar{\mathbf{y}}_i$ would be the same solution as the origin in the $(n-1)$ dimensional LP with $a=i$, hence these types of vertices can be disregarded since they are captured for another $a$. Also, a vertex at the origin, for which constraint \eqref{eq:proof_brs_karma_constr} is not active, cannot be a solution if $a\neq 1$, because a non-null increase in $\bar{\mathbf{y}}_j$ with $j<i$ would still satisfy \eqref{eq:proof_brs_karma_constr} and decrease the cost \eqref{eq:proof_brs_karma_cost}. Nevertheless, if $a= 1$ and the origin satisfies \eqref{eq:proof_brs_karma_constr}, i.e., $k \geq k(j,1)$, then $\mathbf{\bar{y}_j^\star} = \mathbf{e_1}$. Henceforth, we are only interested in the intersections of the boundary hyperplane of  \eqref{eq:proof_brs_karma_constr} and the axes when $k \leq k(j,1)$. Such an intersection with  the axis of coordinate $\bar{\mathbf{y}}_i$ is defined by 
\begin{equation*}
\bar{\mathbf{y}}_i = \frac{( k-k(j,a))}{T(\mathbf{p}_i-\mathbf{p}_a)}.
\end{equation*}
Note that such intersection must satisfy 
\begin{equation}\label{eq:proof_brs_feas_region_a}
0 \leq  \frac{( k-k(j,a))}{T(\mathbf{p}_{i} -\mathbf{p}_a)} \leq 1.
\end{equation}
If $i < a$, \eqref{eq:proof_brs_feas_region_a} it is equivalent to $k(j,a)  \leq k \leq k(j,i)$. Conversely, if $i > a$  \eqref{eq:proof_brs_feas_region_a} it is equivalent to $ k(j,i) \leq  k \leq k(j,a)$. Given that $k(j,a) < k(j,i)$ for $i < a$ and $k(j,a) > k(j,i)$ for $i> a$, then  \eqref{eq:proof_brs_feas_region_a} is equivalent to $ \min(k(j,a),k(j,i)) \leq  k \leq \max(k(j,a),k(j,i))$. For fixed $a$, the cost \eqref{eq:proof_brs_karma_cost} is thus minimized at the intersection with the axis of the coordinate $\hat{j}_a$, given by \eqref{eq:j_hat}.
For $k < k(j,n)$, then the integer choice $j$ is unfeasible, according to Lemma~\ref{lem:feasibility}. Thus, for a fixed $j$, under feasibility, the solution to the LP  is 
\begin{equation*}
\mathbf{\bar{y}_j}^\star := \begin{cases}
	\mathbf{\bar{y}}^\star(j,\hat{a}), \, & k < k(j,1)\\
	\mathbf{e_1}, \, & k \geq k(j,1)
\end{cases},
\end{equation*}
where 
\begin{equation*}
\mathbf{\bar{y}}^\star({j,a}) :=    \frac{k-k(j,a)}{T(\mathbf{p}_{\hat{j}_a}-\mathbf{p}_a)}\mathbf{e_{\hat{j}_a}} -  \left(1-\frac{k-k(j,a)}{T(\mathbf{p}_{\hat{j}_a}-\mathbf{p}_a)} \right)\mathbf{e_a},
\end{equation*}
which can be rewritten as \eqref{eq:y_j_a_star}, and $\hat{a}$ is a component of $\bar{\mathbf{y}}$ that, when eliminated in the $n$ dimensional LP, places a solution along the axes of the space of the remaining variables. Therefore, $\hat{a}$ is given by \eqref{eq:a_hat}.

Now that, for each integer decision $\mathbf{y}^\star = \mathbf{e_j}$, 	a solution to the non-integer component is known and given by $\mathbf{\bar{y}_j}^\star$, the solution to the original MILP is the pair $(\mathbf{y}^\star = \mathbf{e_j},\mathbf{\bar{y}_j}^\star), j \in \{1,\ldots,n\}$ that achieves the lowest cost, which is given by \eqref{eq:singleAgentAverage_cost}. Consider two integer decisions $\mathbf{e_j}$ and $\mathbf{e_i}$ with $i,j\in\{1,\ldots,n\}: i<j$. If both are feasible, i.e., if $k \geq \max(0,\mathbf{p}_i,k(i,n))$, the decision of choosing $\mathbf{e_i}$ over $\mathbf{e_j}$ achieves lower or equal cost when $s/\bar{s} \leq \gamma_{i,j}$, where
\begin{equation*}
\gamma_{i,j} := \frac{T\mathbf{d}(\mathbf{x})^\top(\mathbf{\bar{y}_i}^\star-\mathbf{\bar{y}_j}^\star)}{\mathbf{d}_i(\mathbf{x}_i)-\mathbf{d}_j(\mathbf{x}_j)}.
\end{equation*}
If $k \geq \max(0,\mathbf{p}_i,k(i,n))$, the convention  $\gamma_{i,j} := \infty$ is adopted, yielding \eqref{eq:gamma_ij}. Then $\mathbf{e_{j^\star}}$ is an integer solution if $s/\bar{s}$ is such that $j$ is chosen over every $i \in \{j+1,\ldots,n\}$, i.e., $s/\bar{s} \geq \underline{\gamma}_j$, it is chosen over every $i \in \{1,\ldots,j-1\}$, i.e., $s/\bar{s} \leq  \bar{\gamma}_j $, and $\bar{\gamma}_j < \underline{\gamma}_j , s/\bar{s} < \underline{\gamma}_j$. This proves the sufficiency of the conditions for $j$ to be a solution. The necessity is immediate since, if $ \bar{\gamma}_j < \underline{\gamma}_j , s/\bar{s} < \underline{\gamma}_j$, or {$s/\bar{s} >  \bar{\gamma}_j$}, $\mathbf{e_j}$ cannot be an integer solution. 


\bibliographystyle{IEEEtran}
\bibliography{bib/main.bib,bib/SML_papers.bib}

\end{document}